\documentclass[sn-mathphys-num]{sn-jnl}% Math and Physical Sciences Numbered Reference Style 

\usepackage{graphicx}%
\usepackage{multirow}%
\usepackage{amsmath,amssymb,amsfonts}%
\usepackage{amsthm}%
\usepackage{mathrsfs}%
\usepackage[title]{appendix}%
\usepackage[usenames,dvipsnames]{xcolor}%
\usepackage{textcomp}%
\usepackage{manyfoot}%
\usepackage{booktabs}%
\usepackage{algorithm}%
\usepackage{algorithmicx}%
\usepackage{algpseudocode}%
\usepackage{listings}%

\theoremstyle{thmstyleone}%
\newtheorem{theorem}{Theorem}[section]% meant for sectionwise numbers

\theoremstyle{thmstyletwo}%
\newtheorem{example}{Example}%
\newtheorem{remark}{Remark}%
\newtheorem{corollary}{Corollary}
\theoremstyle{thmstylethree}%
\newtheorem{definition}{Definition}%

    \usepackage{tikz}

\newcommand{\F}{\mathbb{F}}
\newcommand{\z}{\mathbb{Z}}

\newcommand{\Y}{\mathcal{Y}}
\newcommand{\X}{\mathcal{X}}
\newcommand{\HH}{\mathcal{H}}
\newcommand{\PP}{\mathbb{P}}

\newcommand{\bpf}{\begin{proof}}
\newcommand{\epf}{\end{proof}}

\newcommand{\supp}{ \textrm{Supp } }

\newcommand{\rmv}[1]{}

\begin{document}

\title[Algebraic H-LRCs with nested affine subspace recovery]{Algebraic hierarchical locally recoverable codes with nested 
affine subspace recovery}

\author*[1]{\fnm{Kathryn} \sur{Haymaker}}\email{kathryn.haymaker@villanova.edu}

\author[2]{\fnm{Beth} \sur{Malmskog}}\email{bmalmskog@coloradocollege.edu}
\equalcont{These authors contributed equally to this work.}

\author[3]{\fnm{Gretchen} \sur{Matthews}}\email{gmatthews@vt.edu}
\equalcont{These authors contributed equally to this work.}

\affil*[1]{\orgdiv{Department of Mathematics and Statistics}, \orgname{Villanova University}, \orgaddress{\street{800 E. Lancaster Ave}, \city{Villanova}, \postcode{19085}, \state{Pennsylvania}, \country{USA}}}

\affil[2]{\orgdiv{Department of Mathematics \& Computer Science }, \orgname{Colorado College}, \orgaddress{\street{819 N. Tejon St.}, \city{Colorado Springs}, \postcode{80903}, \state{Colorado}, \country{USA}}}

\affil[3]{\orgdiv{Department of Mathematics}, \orgname{Virginia Tech}, \orgaddress{\street{225 Stanger Street}, \city{Blacksburg}, \postcode{24061-1026}, \state{Virginia}, \country{USA}}}

\abstract{
Codes with locality, also known as locally recoverable codes, allow for recovery of erasures using proper subsets of other coordinates. These subsets are typically of small cardinality to promote recovery using limited network traffic and other resources. Hierarchical locally recoverable codes allow for recovery of erasures using sets of other symbols whose sizes increase as needed to allow for recovery of more symbols. In this paper, we 
describe a hierarchical recovery structure arising from geometry in Reed-Muller codes and codes with availability from fiber products of curves.    
We demonstrate how the fiber product hierarchical codes  can be viewed as punctured subcodes of Reed-Muller codes, uniting the two constructions. This point of view provides natural structures for local recovery with availability at each level in the hierarchy. }



\maketitle

\section{Introduction} \label{Introduction}

Error-correcting codes provide ways of encoding information into vectors for storage or communication with redundancy included, so that errors can be detected and recovered by the retriever or recipient.  Not all errors are created equal, however, nor are all errors equally likely to occur across different applications.  For example, symbol erasure (where a vector symbol is erased, leaving a blank coordinate) is easier to detect (and potentially to correct) than symbol changing, because the existence and position of the erasure is known. In cloud storage, a regular event and major concern is that one or more servers will fail catastrophically or be overloaded with requests for data so that they are unable to satisfy additional queries, effectively erasing all data on the server.  Locally recoverable codes are motivated by the desire to facilitate erasure recovery in this setting. The idea is to determine a small set of helper positions for each position $i$ so that  position $i$ may be recovered using only the helper set.  A linear code is said to have locality $r$ if for each coordinate $i$ of a codeword, there is a set of $r$ other helper coordinates, called a recovery or helper set, so that in any codeword, the symbol in position $i$ can be recovered from the symbols in the helper set.

Two structural generalizations that handle multiple simultaneous erasures are availability and hierarchy \cite{Wang_Zhang, barg2017locally, sasidharan2015codes}. A locally recoverable code is said to have availability $t$, referred to as an LRC($t$), if each coordinate has $t$ independent recovery sets.  A locally recoverable code is said to have hierarchical recovery and is referred to as an H-LRC if the recovery set for each coordinate is contained in a larger recovery set that can recover additional erasures beyond what the smaller recovery set can, as depicted in Figure \ref{fig:hierarchy}.
\begin{figure}
    \centering
    \includegraphics[width=2.5in]{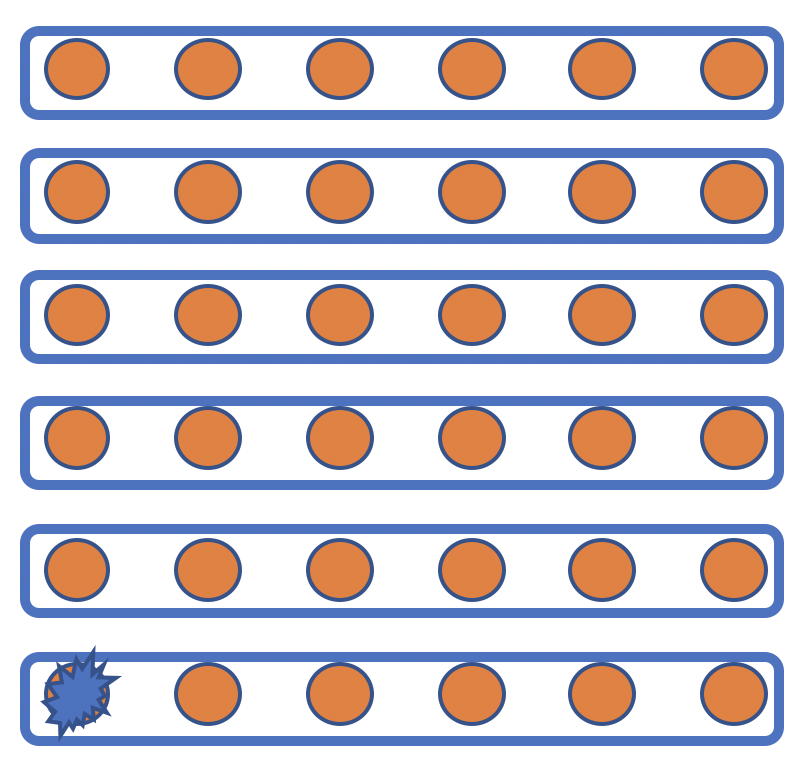}\includegraphics[width=2.5in]{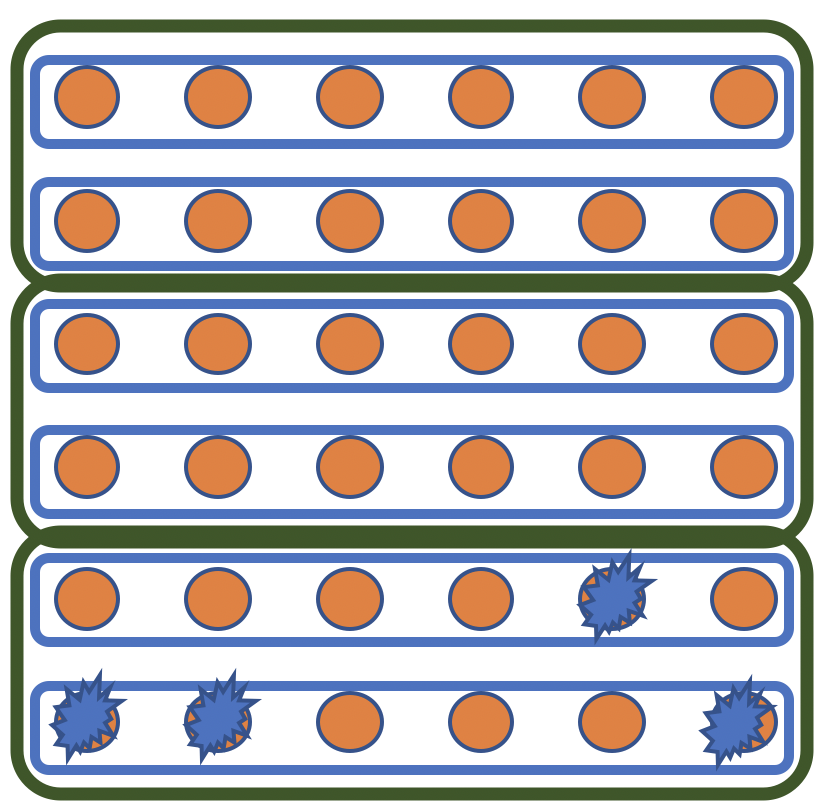}
    \caption{This is a cartoon illustrating a code with 2-level hierarchy.  Small repair groups of size 6 can recover a single erasure, while larger repair groups of size 12 can recover up to 4 erasures.}
    \label{fig:hierarchy}
\end{figure}
This feature can be generalized to several levels of hierarchy.  These two structural generalizations can be combined, giving rise to hierarchical LRC($t$)s. Aside from \cite{ballentine2019codes}, combined hierarchy and availability is a relatively unstudied area, yet many geometric code constructions naturally possess both availability and hierarchy.  Furthermore, geometric constructions often offer great flexibility, allowing hierarchical recovery in many different ways.  In this paper, we examine how hierarchy and availability arise naturally from underlying structure in two related types of geometric codes: Reed-Muller codes and codes from fiber products. 
In \cite{haymaker2018locally}, the authors construct LRC($t)$s from fiber products of $t$ curves.  
We relate these fiber product codes to Reed-Muller codes and %} 
demonstrate how to perform hierarchical recovery for all such codes. 
Though Reed-Muller codes have been well-studied, and the hierarchy arises quite naturally, this approach to hierarchical locality for Reed-Muller codes has not appeared before in the literature.

There has been extensive and interesting work on codes that are optimal with respect to availability bounds.  In one fairly recent example, the authors of \cite{cai2019optimal} construct optimal codes with availability capable of repairing $\delta-1$ erased symbols in the case where each repair set contains one check symbol. 
In \cite{cai2019on-optimal} the authors prove a new bound for locally repairable codes with multiple disjoint repair sets (availability) and present an optimal construction of codes that meet the bound. 
A class of locally recoverable codes with availability that meets the generalized Singleton bound for LRCs is presented in  \cite{garrison2023class}.  
In \cite{silberstein2018anticode}, new classes of LRCs with small  locality ($r=2$ or 3) and high  availability  
are presented, and for many cases, the authors show the codes are alphabet-optimal or Greismer-optimal.

There has also been significant work on H-LRCs, including the following examples, which do not include  availability at each level. 
The authors of \cite{chen2022new} use matrix product codes, concatenated codes, and cyclic codes to construct optimal H-LRCs with lengths either $q^2+q$, $n$ dividing $q-1$, or unbounded. 
In \cite{dukes2023optimal}, the authors give an improved bound for H-LRCs for a special set of parameters and present optimal constructions that meet this bound for those parameters. 
Optimal H-LRCs via generalized Reed-Solomon codes are constructed in \cite{zhang2020constructions}. 
The complete hierarchical locality of the punctured Simplex codes is given in \cite{grezet2021complete}. Note that the codes discussed in \cite{grezet2021complete} also have availability but combined hierarchy and availability is not the topic of that paper.

There are few examples in the literature of LRCs with combined hierarchy and availability.  
 Constructions in \cite{ballentine2019codes} of codes with $H=2$ and availability at each level are given in subsections VIII.D and VIII.E. The difference between the codes in \cite{ballentine2019codes} and the ones here are highlighted in later remarks. 
In \cite{freij2016locally}, the authors give an overview of the connections between matroid theory and linear LRCs with hierarchical availability via examples. They conclude that a level-specific Singleton-type bound is needed, and that remains an open question in this area.

The general constructions in this paper give multiple choices of recovery set at each level of hierarchy (and typically high-availability, all arising naturally from geometry).   A major goal of this paper is to begin the exploration and discussion of codes with structure beyond what has been discussed in the literature--in particular, codes that have availability, hierarchy, and many choices of hierarchical structure (this could be thought of as ``availability of hierarchy"). Instead of devising new codes that have this structure, a major contribution of this paper is to describe how these properties arise naturally in existing, well-studied codes.  The codes in this paper are generally not optimal with respect to existing bounds for codes with availability or hierarchy (e.g. in \cite{ballentine2019codes, sasidharan2015codes})   
 because these bounds do not take into account the additional structure we describe.   
 However, we show that these generally non-optimal codes offer intrinsic versatility in recovery beyond that offered by codes optimal with respect to bounds for either availability or hierarchy. It remains an open question to establish bounds that fully capture this level of flexibility in recovery.

An important contribution of this work is the illumination of nested affine subspaces within existing code families---specifically Reed-Muller codes and codes from fiber products of curves---that can be used for the  purpose of natural hierarchical recovery with high availability. This work is the first time that the combined properties of hierarchy and availability are being exploited in full generality in these code families. Moreover, since many algebraic geometry codes are punctured Reed-Muller codes, the ideas here can be generalized to other structures, resulting in more flexible recovery methods for other existing code families.

This paper is organized as follows. Coding theory preliminaries are contained in Section \ref{Background}. The main families of codes considered in this paper, Reed-Muller codes and fiber product codes, are reviewed here. 
Section \ref{sec:geom_code_intro} contains results on their localities and availability  \cite{haymaker2018locally},  connects these codes concretely and conceptually, and  describes how to obtain basic (2-level) hierarchy from each of these constructions.  Section \ref{sec:hlevelH-LRC} extends this construction to $H$-level hierarchy and demonstrates that these constructions also inherently have availability at each level of recovery.  Section \ref{section:hierarchy and availabililty} uncovers the combined availability and hierarchy inherent in the geometry of Reed-Muller and fiber product codes. The conclusion is found in Section \ref{conclusion}.

\section{Locality, availability, and code constructions} 
\label{Background} 

In this section, we review the notions of locality and availability and set notation to be used throughout the paper. This section contains definitions of terminology of algebraic curves. We also give an overview of the main code constructions that will be needed in later sections. 

\subsection{Background and notation}\label{subsec:background}
We use standard notation from coding theory. The finite field with $q$ elements is denoted by $\F_q$. An $[n,k,d]$ code $C$ over $\F_q$ is a $k$-dimensional $\F_q$-subspace of $\F_{q}^n$ in which any elements differ in at least $d$ coordinates. The set of positive integers is denoted $\z^+$. Given an $[n,k,d]$ code $C$ and $I:= \{i_1,\ldots,i_s\} \subseteq \left\{ 1, \dots, n \right\}$, the corresponding {punctured code} is $\left. C\right|_I := \{(c_{i_1},\ldots, c_{i_s}) : (c_1,\ldots,c_n) \in C\}.$ 

The notion of locality was introduced by Gopalan, Huang, Semitci, and Yekhanin in 2012 \cite{gopalan}. It was generalized to consider multiple failures 
shortly thereafter, as captured in the definition below.

\begin{definition} \cite{Prakash}\label{LRCdef2} Let $n, r, \rho \in \z^+$. A linear code $C$ of length $n$ over $\mathbb{F}_q$ is \textbf{$(r,\rho)$-locally recoverable} if for each $i\in \left\{ 1, \dots, n \right\}$ there exists a punctured code $C_i$ whose support contains $i$, of length at most $r+\rho-1$ and minimum distance at least $\rho$. The support of $C_i$ is called a \textbf{repair group} for position $i$, and $I_i:=\supp C_i \setminus \{i\}$ is called a \textbf{recovery set} for position $i$. We sometimes say that $C$ has $(r, \rho)$ locality.
\end{definition}

Notice that $\supp C_i$ is a repair group for all $j \in \supp C_i$ not just $i$ itself.

Given any $c \in C$ where $C$ has locality  $(r,\rho)$, the symbol $c_i$ in position $i$ may be recovered using at most $r$ symbols indexed by elements of $I_i$, assuming no more than $\rho-1$ erasures have occurred. If $r$ symbols are not available due to erasures, the codeword symbols indexed by $\left\{ 1, \dots, n \right\} \setminus \{i\}$ may be used to recover $c_i$.  This idea could be seen as foreshadowing the notion of hierarchical recovery: if too many erasures prevent recovery by a smaller set ($I_i$ in this case), recovery is attempted using the larger set $\left\{ 1, \dots, n \right\} \setminus \{ i \}$. Of course, we are interested in the setting in which a there is another set slightly larger than $I_i$ that can be used for recovery rather than immediately considering all remaining coordinates indexed by $\left\{ 1, \dots, n \right\} \setminus \{ i \}$. This is known as hierarchical locality, and the formal definition will be given in Section \ref{fiber_product_construction_section}.

If too many positions within a recovery set become unavailable, local recovery may not be possible. This leads to what is known as the availability problem.  One way to address this problem is to construct multiple recovery sets for each position of the codeword. A code is said to have availability $t$ if each coordinate has $t$ independent recovery sets. The concept of availability was introduced in \cite{Wang_Zhang}. More formally, we have the following definition provided by Ballentine, Barg, and Vl\u{a}dut. 

\begin{definition}  \cite{ballentine2019codes} \label{LRCtdef2}
Let $n, t, r_j, \rho_j \in \z^+$ for all $j\in\left\{ 1, \dots, t \right\}$. A linear code of length $n$ over $\mathbb{F}_q$ is \textbf{ $(r_j,\rho_j)_{1\leq j\leq t}$-locally recoverable with availability $t$} if, for each $i\in\left\{ 1, \dots, n \right\}$ and each $j\in\left\{ 1, \dots, t \right\}$, there exists a punctured code $C_{i,j}$ with support positions $\{i\}\cup I_{i,j}$ such that dim$(C_{i,j})\leq r_j$, there are at least dim$(C_{i,j})$ linearly independent positions in $I_{i,j}\setminus \cup_{k\in[t]\setminus\{j\}} I_{i,k}$, and $d(C_{i,j})\geq \rho_j$; that is, $C_{i,j}$ is a $(r_j, \rho_j)$-locally recoverable code.  The set $I_{i,j}$ is known as the $j$-th recovery set for position $i$.
\end{definition}

The following definitions related to algebraic curves over finite fields are from Chapter 13 of \cite{huffman2010fundamentals}, and we refer the reader to that reference as well as \cite{hoholdt1998algebraic, tsfasman2007algebraic} for additional details and examples.  For a thorough treatment of algebraic curves and their associated function fields, please see Appendix B of \cite{stichtenoth2009algebraic}.

Let $\mathbb{F}$ be a field. The \textit{$n$-dimensional affine space over} $\mathbb{F}$, sometimes denoted $\mathbb{A}^n(\F)$, is the $n$-dimensional vector space $\mathbb{F}^n$, where the elements of $\mathbb{F}^n$ are called \textit{points}. The \textit{$n$-dimensional projective space} $\mathbb{P}^n(\mathbb{F})$ is the set of all equivalence classes
\[\left(\{[x_1:x_2:\dots :x_{n+1}]:x_i\in \F \textrm{ for all }i\}\setminus\{[0:0:\dots :0]\}\right)/ \sim,
\]
where  $[x_1, x_2, \ldots, x_{n+1}]\sim [y_1, y_2, \ldots, y_{n+1}]$ if and only if $x_i=\lambda y_i$ for all $i$, where $\lambda \in \mathbb{F}\setminus\{0\}$. We generally call $P=[x_1: x_2: \cdots : x_{n+1}]\in \mathbb{P}^n(\mathbb{F})$ a \textit{point at infinity} if $x_{n+1}=0$. \textit{Affine points} in $\mathbb{P}^n(\mathbb{F})$ are those points not at infinity. Each affine point in $\mathbb{P}^n(\mathbb{F})$ can be uniquely represented as $[x_1:x_2: \cdots: x_n:1]$, giving a natural identification with the points of $\mathbb{F}^n$.  The \textit{projective line over} $\mathbb{F}$ is $\mathbb{P}^1(\mathbb{F})$ and the \textit{projective plane} over $\mathbb{F}$ is $\mathbb{P}^2(\mathbb{F})$. Let $X_1, X_2, \ldots, X_n$ be indeterminates. A polynomial $f$ in $\F[X_1, X_2, \ldots, X_n]$ is \textit{homogeneous of degree $d$} if every term of $f$ is of degree $d$. A polynomial can be made homogeneous by introducing a new variable $X_{n+1}$ and multiplying each term by an appropriate power of $X_{n+1}$ to make all terms have degree $d$. We denote by $f^H$ the homogenized polynomial of $f$. This process gives a one-to-one correspondence between polynomials in $n$ variables of degree at most $d$ and homogeneous polynomials of degree $d$ in $n+1$ variables. 
The $\F$-points of an \textit{affine plane curve} $\X_f$ defined by $f(x,y)\in \F[X,Y]$ a non-constant polynomial is 
\[\X_f(\F)=\{(x,y)\in \mathbb{F}^2: f(x,y)=0\}.\] 
The associated function field $\F(\X_f)$ is the field of rational functions $\F(x,y)$, where the relation between $x$ and $y$ is given by $f(x,y)=0$.
A \textit{projective plane curve} $\X_f$ is the set of projective points $[x:y:z]\in \mathbb{P}^2(\F)$ such that $f(x,y,z)=0$, where $f(x,y,z)$ is a non-constant homogeneous polynomial. If a projective curve is defined by a polynomial $f^H$, it is called the \textit{projective closure} of $\X_f(\F)$.  
Curves defined over $\F$ can also have points over an extension field $\mathbb{E}$, in which case the notation $\X_f(\mathbb{E})$ is used. 
The partial derivative of a function is defined as in calculus, and a point $(y_1, y_2, \ldots, y_n)$ 
on an affine (resp., projective) curve $\X_f(\F)$ is called \textit{singular} if every partial derivative of $f$ equals 0 at the point. A curve that has no singular points is called \textit{smooth}. 

To define an evaluation code with evaluation points coming from a curve, we require a vector space of functions that can be evaluated at the points. Let $p$ be a homogeneous polynomial of positive degree that defines a projective plane curve $\X_p$ over $\F$. Define the \textit{field of rational functions on} $\X_p$ over $\F$, denoted by $\F(\X_p)$, as the collection of equivalence classes of rational functions $\frac{g}{h}$, where $g$ and $h$ are homogeneous of equal degree, with $p\nmid h$, and where $\frac{g}{h}\sim \frac{g+ap}{h+bp}$ for any $a,b$ homogeneous polynomials of appropriate degree.

 In \cite{barg2015locally}, the authors construct locally recoverable codes with availability $t=2$ based on fiber products of curves and propose a group-theoretic perspective on the construction, whereby a curve can sometimes be expressed as a fiber product of its quotient curves by certain subgroups of the automorphism group of the curve; see \cite{barg2017locally} for an extended version. In \cite{haymaker2018locally}, the authors give a closely related construction of codes with availability  $t$ recovery sets for any $t\geq 2$ based on fiber products of $t$ curves, including a different method of designing the LRC($t$)s, bounding the parameters for these codes, and constructing examples based on the generalized Giulietti-Korchmaros curves, the Suzuki curves, and fiber products of Artin-Schreier curves proposed by van der Geer and van der Vlugt.  For an additional reference on  codes from fiber products of curves, we recommend \cite{chara2023minimum}.
 In the next subsection, we review these constructions in the context of the framework above in preparation for taking a hierarchical perspective.

\subsection{Reed-Muller and fiber product codes}
In this section, we introduce Reed-Muller codes and codes from fiber products of curves, then describe a framework that unites these codes as locally recoverable codes with availability. In the next section, we will see that it unifies the hierarchical recovery structure of these two families as well.

Recall that Reed-Muller codes are evaluation codes, defined as follows. 

\begin{definition}
The \textbf{($q$-ary) Reed-Muller code} is a linear code over $\mathbb{F}_q$ formed by evaluations of polynomials in $\F_q[x_1, \ldots, x_m]$ of total degree at most $v$ at points in $\F_q^m$, and is denoted by $\mathcal{RM}_q(v,m)$; that is, 
$$
\mathcal{RM}_q(v,m)= \left\{ \left( f(P_1), \dots, f(P_n) \right): f \in \F_q[x_1, \ldots, x_m]_{\leq v} \right\}
$$
where $n=q^m$ and $\F_{q}^m = \left\{ P_1, \dots, P_n \right\}$.
\end{definition} 

In thinking about the geometry of $\F_{q}^m$, we will make use of linear subspaces and their translations.  The following definition makes these objects precise.

\begin{definition}
    Viewing $V=\mathbb{F}_q^m$ as a vector space over $\mathbb{F}_q$, define an \textbf{affine subspace} $A$ of $V$ to be a coset of a subspace  $S$ of $V$, 
    \[ A=\vec{v}+S.\]
    When $S$ has dimension $r$, the set $A$ is also called an \textbf{$r$-flat}. 
\end{definition}

Local properties of Reed-Muller codes were considered by Yekhanin in \cite{yekhanin2012locally} for the purpose of local decodability.  If $v< q-1$, the code $\mathcal{RM}_q(v,m)$ has local recovery with $r=q-1$ and availability $\frac{q^m-1}{q-1}$. To see this, note that the value of any polynomial of total degree at most $v$ at a point $P_i$ in $\F_q^m$ can be determined by the value of that polynomial on the other $q-1$ points on any of the $\frac{q^m-1}{q-1}$ distinct $\mathbb{F}_{q}$-rational lines through $P_i$. It is important to note that because two lines intersect in at most one point, and in this case that point is $P_i$, these lines with $P_i$ removed satisfy the disjoint repair group property. In this way, we see the inherent geometric structure giving rise to LRCs with availability. 

Now, we turn our attention to another code design in which the underlying geometry gives rise to availability. In what follows, we build on Section \ref{subsec:background}, but do use some more general and technical definitions.  These objects are defined carefully in the references \cite{stichtenoth2009algebraic, huffman2010fundamentals, hoholdt1998algebraic, tsfasman2007algebraic}.  However, we will attempt to provide some intuition on the more general and technical terms as we proceed so that the main ideas may be clearer to the reader. Further, the codes considered here are  defined in full generality in \cite{haymaker2018locally} and discussed at length in \cite{chara2023minimum}, where space is devoted to developing a more concrete view. 

Intuitively, an algebraic curve is the one-dimensional vanishing set of some polynomials defined over a given field.  The curve is smooth if it does not have any self-crossings or cusps, i.e. it has a one-dimensional tangent space at each point.  In this context, projective means that we include any ``points at infinity" on each curve, which can be made algebraically precise by homogenizing the defining equations for the curve and looking for solutions in projective space. A rational map of curves is a map defined by rational functions.  Let $ \mathcal{Y}, \mathcal{Y}_1, \mathcal{Y}_2, \ldots, \mathcal{Y}_t$ be smooth, projective algebraic curves over $\mathbb{F}_q$ each with a rational, separable map $h_j: \mathcal{Y}_j \rightarrow \mathcal{Y}$ of degree $d_{h_j}$ for each $j\in\{1,2,\dots t\}$.  We note that separability is a more technical concept, but that a map is certainly separable if the characteristic of $\F_q$ does not divide the degree of the map.

Each projective algebraic curve $\X$ has an associated field of rational functions defined over $\F_q$, denoted $\F_q(\X)$.  We assume two further technical conditions on our curves, namely that $\F_q$ is the full field of constants within each of the associated function fields $\F_q(\mathcal{Y}_i)$, and that the extensions $\F_q(\mathcal{Y}_i)/\F_q(\mathcal{Y})$ are linearly disjoint (see \cite{chara2023minimum} for more on this). Let $\mathcal{X}$ be the fiber product $\mathcal{Y}_1\times_{\mathcal{Y}}\mathcal{Y}_2\times_{\mathcal{Y}}\dots \times_{\mathcal{Y}}\mathcal{Y}_t$.  For a full discussion of fiber products, see \cite{Gortz2010}. Intuitively, the points of $\X$ can be thought of as $t-$tuples of points, one from each of the curves $\Y_i$, where each point in the tuple maps to the same point on $\Y$.  That is, points of $\X$ are of the form $(P_1,P_2,\dots, P_t)$ where $h_1(P_i)=h_2(P_2)=\dots=h_t(P_t)$. Though it is not immediately obvious, this construction results in a curve $\X$. Let the natural projection map $g_j: \mathcal{X} \rightarrow  \mathcal{Y}_j$ have degree $d_{g_j}$ for  each $j$, $1 \leq j \leq t$, and define the rational, separable map $g=g_j\circ h_j: \mathcal{X} \rightarrow \mathcal{Y}$ (for any $j$) of degree $d_g$. For each $j$, define the curve \[\tilde{\mathcal{Y}}_j= \mathcal{Y}_1\times_{\mathcal{Y}} \cdots \times_{\mathcal{Y}}\mathcal{Y}_{j-1}\times_{\mathcal{Y}}\mathcal{Y}_{j+1}\times_{\mathcal{Y}}\cdots \times_{\mathcal{Y}} \mathcal{Y}_t.\]
More generally (for later use) we define $\tilde{\Y}_{A}$ for any set $A\subseteq\left\{ 1, \dots, t \right\}$ to be the fiber product over $\Y$ of all $\Y_i$ with $i\not\in A$. Then $\mathcal{X}=\mathcal{Y}_j\times_{\mathcal{Y}}\tilde{\mathcal{Y}}_j$. Denote the associated natural maps by
\[\tilde{g}_j: \mathcal{X} \rightarrow \tilde{\mathcal{Y}}_j \hspace{.25in} \textrm{ and } \hspace{.25in} \tilde{h}_j:\tilde{\mathcal{Y}}_j\rightarrow \mathcal{Y}.\] The degree of $\tilde{g}_j$ must be equal to $d_{h_j}$.   Intuitively, the map $\tilde{g}_j$ ``forgets" the information coming from the curve $\Y_j$ while retaining the data of the fiber product that come from the other factors, meaning the curves $\mathcal Y_i$ with $i \neq j$.

The function field $\F_q(\X)$ is isomorphic to the compositum of the function fields $\F_q(\Y_i)$, where the function field $\F_q(\Y)$ is embedded into each $\F_q(\Y_i)$ as induced by the map $h_i$.    
For ease of exposition, we identify each function field with its image inside $\F_q(\X)$, so for each $i$, $$\F_q(\Y)\subseteq \F_q(\Y_i)\subseteq \F_q(\X).$$ 

This framework sets the stage for a code design reminiscent of algebraic geometry codes, though harnessing the fiber product structure to facilitate local recovery with high (or bespoke) availability.  Let $y_j \in \F_q(\mathcal{Y}_j)$ be a primitive element of $\F_q(\mathcal{Y}_j)/\F_q(\mathcal{Y})$ where $y_j$ is the root of a degree $d_{h_j}$ polynomial with coefficients in  $\F_q(\mathcal{Y})$.  This yields \[\F_q(\mathcal{Y})(y_1,y_2,\dots,y_t)\cong\F_q(\mathcal{X}).\]

A divisor on a curve is a formal sum of points on the curve.  Let $D_j$ be the principal divisor of the function $y_j\in \F_q(\mathcal{X})$.  Let $D_j=D_{j,+}-D_{j,-}$, where  $D_{j,+}=(y_j)_0$ is the zero divisor of $y_j$, and $D_{j,-}=(y_j)_{\infty}$ is the pole divisor of $y_j$.  Let $d_{y_j}$ be the degree of $y_j:\mathcal{Y}_j\rightarrow\mathbb{P}^1$. Then, if $y_j$ is viewed as a function $y_j:\mathcal{X}\rightarrow \mathbb{P}^1$, its degree is $d_{g_j}d_{y_j}$, which is equal to $\textrm{deg}(D_{j,-})$.  
Now choose simultaneously a divisor $D$ on $\mathcal{Y}(\F_q)$ and a set $S$ of points in $\mathcal{Y}(\F_q)$ as follows. Let $\tilde{D}=\sum_{j=1}^t g(D_{j,-})$, so $\supp(\tilde{D})$ consists of all points of $\mathcal{Y}$ which, for some $j$, are the image under $g$ of a point on $\mathcal{X}$ at which the function $y_j$ has a pole.  Then choose $D$ an effective divisor on $\mathcal{Y}(\F_q)$ of degree $\deg(D)=l$ and $S=\{P_1, \dots, P_s\} \subseteq \mathcal{Y}(\F_q)$ so that the following conditions are satisfied: 
\begin{itemize} 
 \item $|g^{-1}(P_i)\cap\mathcal{X}(\F_q)|=d_g,$ for all $i\in[s]$, 
 \item $S \cap \supp (\tilde{D}) = \emptyset$,  
 \item $S\cap \supp(D) = \emptyset$, 
 \item  $l< s$. 
 \end{itemize} 

Let $\mathcal{L}(\Y,D)$ be the Riemann-Roch space of the divisor $D$ on the curve $\Y$.  See \cite{stichtenoth2009algebraic} for a discussion of Riemann-Roch spaces and their properties.  Let $\ell(D)$ be the dimension of $\mathcal{L}(\Y,D)$ as an $\F_q$-vector space. Let $\{f_1, f_2,\dots, f_m\}$ be a basis for $\mathcal{L}(\Y,D)$, so $m=\ell(D)$. These functions are naturally in $\F(\Y)$, and we also consider them to be functions in $\F(\X)$ through the natural inclusion.  Note that since $l<s$, each non-zero function in $\mathcal{L}(\Y, D)$ will be non-zero when evaluated at some point in $S$. 
Then set 
$$B=g^{-1}(S)$$
so that $n:=|B|=d_gs$. Order the points in $B$ and denote them as $\{Q_1, Q_2, \ldots, Q_n\}$. Let %$V=$
\[V=\textrm{Span}\left\{ f_ky_1^{{e_1}}\cdots {y_t}^{e_t}: \begin{array}{l} e_i\in\mathbb{Z}, \\ 0\leq e_j\leq d_{h_j}-\rho_j \forall j\in \left\{ 1, \dots, t \right\}, k\in [m] \end{array} \right\}. \] 
Define
$C(D,B):=Im (ev_B)$ where 
$$
\begin{array}{llll}
ev_B: &V &\rightarrow &\F_q^{n} \\
&f& \mapsto & \left( f \left( Q_i\right) \right)_{i\in \left\{ 1, \dots, n \right\}}.
\end{array}
$$The following is a restatement of results from \cite{haymaker2018locally}, using the above terminology, in preparation for an investigation of hierarchy.
\begin{theorem}[\cite{haymaker2018locally}]\label{LRCtheorem}
Let curves $\{\mathcal{Y}_j\}_{j\in\left\{ 1, \dots, t \right\}}$, $\mathcal{Y}$, maps $\{h_j:\mathcal{Y}_j \rightarrow\mathcal{Y}\}_{j\in\left\{ 1, \dots, t \right\}}$, a divisor $D$ on $\mathcal{Y}(\F_q)$, and sets $S\subseteq\mathcal{Y}(\F_q)$ and $B=g^{-1}(S)$ be all as described above, where $l=\textrm{deg}(D)\leq |S|$ and the quantity $d$ below is positive.  For $i\in\{1,2,\dots, n\}$ and $j\in\{1,2,\dots, t\}$, set $C_{i,j}$ to be the punctured code within $C(D,B)$ with support $\{k:Q_k\in \tilde{g}_j^{-1}(\tilde{g}_j(Q_i))\}$. Then the code $C(D,B)$ is an LRC($t$) with 
\begin{itemize}
\item length $n=|B|=d_g|S|$, 
\item dimension $k=\ell(D)(d_{h_1}-\rho_1+1)(d_{h_2}-\rho_2+1) \cdots (d_{h_t}-\rho_t+1) $
\item minimum distance $d \geq n-ld_g-\sum_{j=1}^t\left( d_{h_j}-\rho_i\right)\left(d_{g_j}d_{y_j}\right)$, and 
\item locality $r_j=d_{h_j}-\rho_j+1$ for $1\leq j\leq t$.
\end{itemize}

\end{theorem}

Next, we provide a unified framework that will aid the investigation of hierarchical locally recoverable codes from Reed-Muller and fiber product codes.

\subsection{Relating fiber product codes and Reed-Muller codes}

Locally recoverable codes with many recovery sets arising from fiber product constructions were presented in great generality in \cite{haymaker2018locally}.  Though this generality is valuable, many examples of these codes can be understood much more simply.  In these cases, the availability and hierarchical recovery scheme that will be outlined is essentially the same as the natural geometric availability and hierarchical recovery scheme that we will present  for generalized Reed-Muller codes. We will see that this perspective orients fiber product codes in the larger landscape of techniques to increase the rate of Reed-Muller codes and punctured Reed-Muller codes while still retaining availability.

Let all notation be as in Theorem \ref{LRCtheorem}.  Assume that $$\mathcal{L}(\mathcal{Y},D) \subseteq \mathbb{F}_q[z_1,z_2,\dots, z_{\mu}]_{\leq \psi}$$ so that  
 the maximum total degree of all functions in $\mathcal{L}(\mathcal{Y},D)$ is at most $\psi$. Then take
\[V \subseteq\left< z_1^{a_1}\cdots z_{\mu}^{a_{\mu}}y_1^{{e_1}}\cdots {y_t}^{e_t}: 
\begin{array}{l} a_i, e_i\in\mathbb{Z}, 0\leq e_j\leq d_{h_j}-\rho_j \forall j\in \left\{ 1, \dots, t \right\},\\ a_1+\dots+a_{\mu} \leq \psi \end{array} \right>. \]The curve $\mathcal{X}$ abstractly exists in a space which is the product of many projective spaces. Consider a collection $S$ of affine $\F_q$-rational points on $\mathcal Y$, meaning $S\subseteq \mathcal Y ( \F_q)$ may be thought of as $$S \subseteq \F_q^{\mu}.$$  Take $B=g^{-1}(S) \subseteq \mathcal X(\F_q)$ which we identify with points in $\mathbb{F}_q^{\mu+t}$ so that the  
evaluation points 
for the code $C(B,D)$ satisfy $$B \subseteq \mathbb{F}_q^{\mu+t};$$ more precisely, the function $$\phi: \mathcal X \rightarrow \mathbb{F}_q^{\mu+t}$$ gives 
$$(z_1,z_2,\dots,z_{\mu}, y_1,y_2, \dots, y_t) \mapsto \mathbb{F}_q^{\mu+t}.$$The codewords of $C(B,D)$ are then the evaluations of polynomials in $\mathbb{F}_q[z_1,z_2,\dots,z_{\mu}, y_1,y_2, \dots, y_t]$ on points $S$ of the fiber product curve $\mathcal{X}$.  Thus the fiber product code $C(D,B)$ is a subcode of the Reed-Muller code $RM_q(r,\mu+t)$ for some $r$, punctured to the set $\phi(g^{-1}(S))$ where $r\leq\psi+\sum_{j=1}^t (d_{h_j}-\rho_j)$: 
$$C(B,D) \subseteq RM_q(r,\mu+t)\mid_{\phi(g^{-1}(S))}.$$
When viewed from this perspective, the $t$ recovery sets of a position in a $t$-fold fiber product code correspond to the intersections of lines parallel to the $y_1,\dots, y_t$ coordinate axes with the curve.  In a Reed-Muller code, each line through a point gives a recovery set for the position of that point; in the fiber product case we limit ourselves to particular lines yielding known-cardinality intersections with the evaluation set.  By doing this, we are able to use functions of larger total degree than would be possible when using a standard Reed-Muller code.  In particular, we are able to let each $y_k$ have degree up to $d_{h_k}-2$, without any limit on the total degree of the function. Consequently, the more general fiber product codes yield larger rates than would be possible for the punctured Reed-Muller code.

This approach is one of several ways to define codes with larger rates. One way of increasing rate is lifting, pioneered by Guo, Kopparty, and Sudan in \cite{guo2013new}.  Their lifted Reed-Solomon codes make use of the observation that there are generally monomials of total degree $\geq q-1$ in $\mathbb{F}_q[x_1,\dots, x_m]$ that reduce to degree $<q-1$ univariate polynomials on every line in $\F_q^m$.  Adding these monomials to a lower degree Reed-Muller code greatly increases the rate without losing the very high availability.  In \cite{PartiallyLifted}, the authors create partially lifted codes, which increase the rate further by limiting the degree condition to a subset of lines as well as including non-monomial functions that meet the degree condition on lines.  In \cite{lopez2021hermitian}, the authors consider lifted codes where the evaluation points are $\F_{q^2}$-points of the Hermitian curve $\HH_q$ in the plane.  All non-tangent lines with $q+1$ points of intersection with the affine patch of $\HH_q$ yield recovery sets for this code, and the degree allowed for the univariate polynomials on lines is restricted by the size of the intersection with the curve.  

In fiber product codes, we use a subset of curve points as the evaluation sets, but the lines corresponding to recovery sets have been severely limited to those parallel to many coordinate axes. The degree allowed for the restricted polynomials on these lines is again limited by the cardinality of the intersections of the lines with the curve, which in the case of the fiber product construction is the degree of the map $h_k$ for some $1\leq k\leq t$.  This allows for many more monomials than bounding total degree as is the case in the Reed-Muller code.  

In the next sections, we demonstrate how these constructions and modifications naturally give rise to hierarchical local recovery with availability. 

\section{Two-level hierarchical recovery of Reed-Muller and generalized fiber product constructions} \label{fiber_product_construction_section}

\label{sec:geom_code_intro}

In the previous section, we observed how lines and fiber products of curves can provide many disjoint recovery sets for a codeword coordinate. Another concept for addressing the availability problem has been advanced in the form of hierarchical locality \cite{ballentine2019codes, giacomoHierarchy,freij2016locally,  grezet2021complete,  sasidharan2015codes}. The idea is that a position and its recovery set could form a locally recoverable code with smaller locality and local minimum distance, offering two nested recovery sets.  The larger recovery set would be used if the smaller set was insufficient to recover the given erasures. This notion captured in the next definition. 
\begin{definition} [\cite{sasidharan2015codes}]\label{H-LRCdef}
Let $n, n_1, n_2, s_1, s_2, \delta_1, \delta_2\in\z^+$ with $n_2<n_1$, $s_2\leq s_1$, and  $\delta_2<\delta_1$.  A linear code $C$ of length $n$ is said to have \textbf{hierarchical locality} with parameters $((n_1, s_1, \delta_1), (n_2,  s_2, \delta_2))$ if for each $i\in\left\{ 1, \dots, n \right\}$ there exists a punctured code $C_i$ of length $n_1$, dim$(C_{i})\leq s_1$, and minimum distance $d(C_{i})\geq \delta_1$ with $i$ in the support of $C_{i}$ such that $C_{i}$ is an $(s_2,\delta_2)$-locally recoverable code, where each local repair group has size $n_2$.

\end{definition}

Note that up to $\delta_2-1$ erasures 
among the support of any local repair group within $C_{i}$ can be locally recovered using the local recovery process for $C_{i}$.  Further erasures up to $\delta_1-1$ total can be recovered using all the coordinates of $C_{i}$. Hence, there are two levels of hierarchy present: 
\begin{enumerate}
\item one level using at most $n_2-\delta_2+1$ symbols from a recovery set for the symbol within $C_{i}$ if there are fewer than $\delta_2$ erasures and 
\item another larger one using at most $n_1-\delta_1+1$ symbols of $C_{i}$ if there are between $\delta_2$ and $\delta_1-1$ erasures. 
\end{enumerate}

\subsection{Hierarchical recovery of Reed-Muller codes}

The geometric definitions of Reed-Muller codes give them built-in nested structure in the form of intersections of the evaluation set with lines, planes, and hyperplanes within the ambient space, as depicted in Figure \ref{fig:ReedMullerHierarchy}.  
\begin{figure}
    \centering
    \includegraphics[height=3.5in]{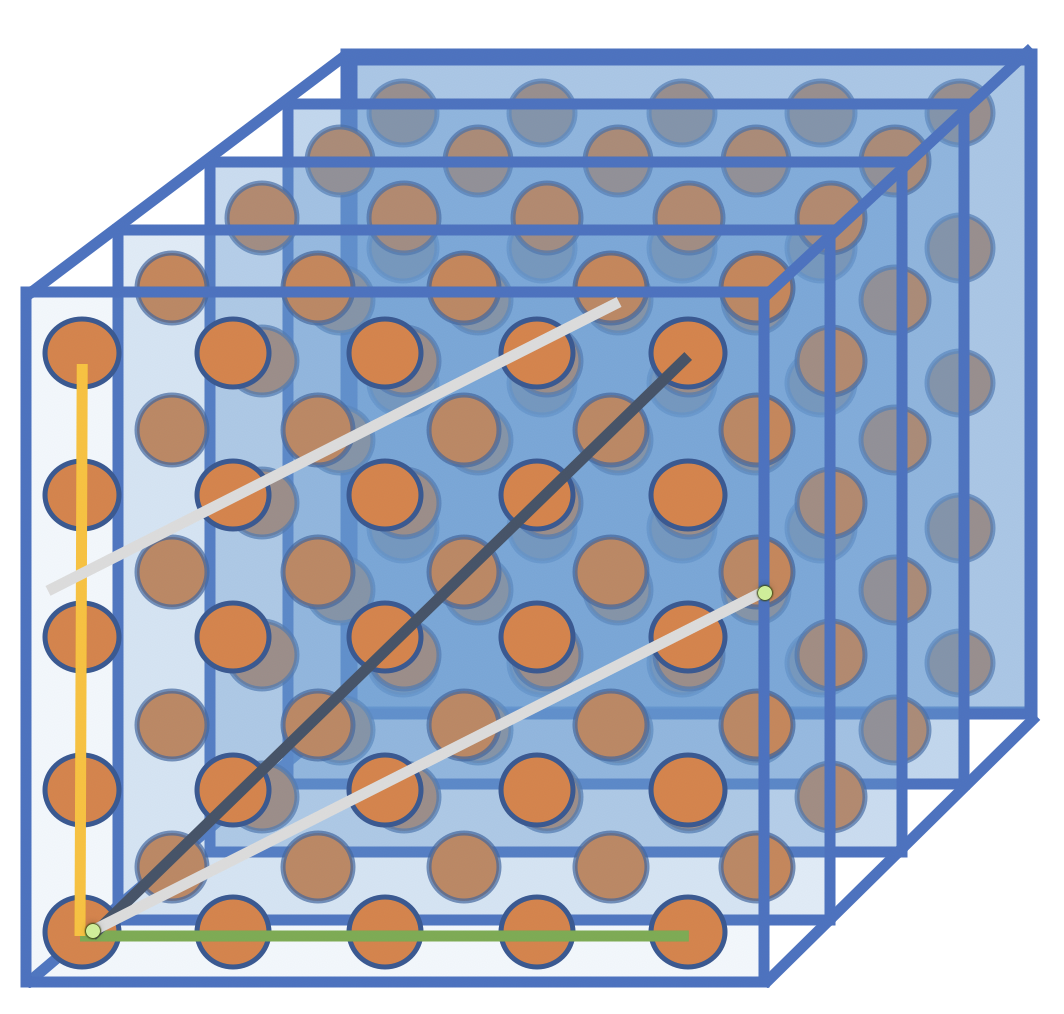}
    \caption{This cartoon illustrates how some hierarchical repair groups naturally arise in  the evaluation points of $RM_5(v,3)$ for $v\leq 3$.  The positions corresponding to points of any $\F_5$-rational line passing through the given point can act as a repair group for local recovery, with a middle code arising from any plane containing the line.  Here we see one way to partition the space into planes, the points of which correspond to the positions in the support of middle codes. For each point in these planes, there are six lines in the plane passing through the point (four are illustrated here), any of which can be chosen to give the positions in the support of the lower code for the chosen position.}
    \label{fig:ReedMullerHierarchy}
\end{figure}We see how to use this structure for basic (2-level) hierarchical recovery in both settings.

\begin{theorem}
    The Reed-Muller code $\mathcal{RM}_q(v,m)$ has hierarchical locality with parameters \[ \left(\left(q^2, \binom{v+2}{2}, q^2 -qv\right),\left(q, v+1,q-v\right)\right).\] 
\end{theorem}
    
\begin{proof}   
Recall that the code $C=\mathcal{RM}_q(v, m)$, where $v\leq q-1,$ has the following parameters: 
\[ n=q^m, k={v+m\choose m}, d=(q-v)q^{m-1}.  \]
Consider the H-LRC  formed by taking planes and lines, respectively, through a given position. 
We  have that any plane containing $P_i$ can be the support for a middle code $C_i$, with local recovery inside the plane given by any partition of the plane into parallel lines.  If $v\leq q-\delta_2$, then $C_{1_i}$ is a $(q^2, q-\delta_2+1,\delta_2)$-locally recoverable code, and itself has dimension dim($C_{1_i})\leq \binom{q-\delta_2+2}{2}$ and minimum distance $d(C_{1_i})\geq q\delta_2$.  

Put in terms of $v$, the code $C=\mathcal{RM}_q(v, m)$ has hierarchical locality with parameters $((n_1, s_1, \delta_1), (n_2, s_2, \delta_2))$, where 
\begin{align*}
n_1& = q^2,    
    s_1={v+2\choose 2}, \delta_1=q^2-qv\\
    n_2 & = q,  
    s_2=v+1,  
    \delta_2=q-v.
\end{align*}
Indeed, the middle code for a fixed position $c_i=(x_1, \ldots, x_m)_i$ consists of evaluations of all points in a plane containing the point $c_i$. Therefore $s_1\leq q^2$, the number of points on a plane. The evaluations of the points on a plane form an LRC with parameters $(v+1, q-v)$, since there is a restriction to evaluations of a line contained in the plane. The middle code is isomorphic to  $\mathcal{RM}_q(2, q-2)$, and the bottom code is isomorphic to $\mathcal{RM}(1, q-2)$.  
\end{proof}

\begin{example} [Two-level hierarchical recovery using a Reed-Muller code]
    Consider $\mathcal{RM}_7(5,3)$, with $H=2$. Then $\mathcal{C}_2\subseteq \mathcal{C}_1\subseteq \mathcal{C}$, where the parameters of $\mathcal{C}$ are $[343, 56, 98]$, and the H-LRC parameters are $[49, 21, 14]$, and $[7,6,2]$, respectively. 
\end{example}

\begin{remark} 
Variations of the above two-level  hierarchy construction above can be obtained using a code $\mathcal{RM}(v,m)$ where the supports used for $n_1, n_2$ are defined by nested $\ell_1$- and $\ell_2$-flats in the geometry, with $m-1>\ell_1>\ell_2>1$. The hierarchical parameters are 
\[ \left(\left(q^{\ell_1}, \binom{v+\ell_1}{\ell_1}, (q-v)q^{\ell_1-1}\right),\left(q^{\ell_2}, \binom{v+\ell_2}{\ell_2}, (q-v)q^{\ell_2-1}\right)\right)\]
giving added flexibility in the construction.
\end{remark}

\subsection{Hierarchical recovery of fiber product codes}

In this subsection, we return to fiber product codes, examining the family's ability to support local recovery. In the fiber product construction of \cite{haymaker2018locally}, we see that the $j$-th recovery set for position $i$ consists of the positions corresponding to points on the curve that all share the same images under the function $g$ and each $y_k$ except $k=j$.  To provide the appropriate setting to discuss hierarchical recovery, consider a fiber product code $C(D,B)$ and define $C_{i,j}$ to be the punctured code of $C(D,B)$ with support $\{k:Q_k\in \tilde{g}_j^{-1}(\tilde{g}_j(Q_i))\}$. In other words, the $i$-th coordinate corresponds to a point $Q_i$ on the fiber product curve.  The $j$-th recovery set corresponds to a factor curve $\mathcal{Y}_j$ of the fiber product, and element $y_j$ in the function field $\F_q(\mathcal{X})$.  The support of $C_{i,j}$ will contain all coordinates corresponding to points $Q$ such that $y_k(Q)=y_k(Q_i)$ for all $k\neq j$ and $g(Q)=g(Q_i)$.  Local recovery of coordinate $i$ with $C_{i,j}$ is accomplished by observing that on the points with indices in the support of $C_{i,j}$, any function in $V$ restricts to a polynomial of degree at most $d_{h_j}-\rho_j$ in $y_j$, which can be interpolated at any missing points given the value of the polynomial on at least $d_{h_j}-\rho_j+1$ points. Thus each $C_{i,j}$ is a subcode of a punctured Reed-Solomon code. Any $\rho_j-1$ erasures within the support of $C_{i,j}$ are recoverable using $C_{i,j}$.  We formalize the $C(D,B)$ as a two-level H-LRC and record its parameters in the following theorem. 

\begin{theorem}\label{H-LRCTheorem}
Let $t$ be a positive integer with $t\geq 2$. Let $C=C(D,B)$ be an LRC($t$) constructed from a fiber product of curves as in Theorem \ref{LRCtheorem}.  Choose any $j,k\in\{1,2,\dots, t\}$ with $j\neq k$. Then $C$ has hierarchical locality with $n_1= d_{h_j}d_{h_k}$, $n_2= d_{h_j}$, $s_1=(d_{h_j}-\rho_j+1)(d_{h_k}-\rho_k+1)$, $s_2 = d_{h_j}-\rho_j+1$, $\delta_1 = \rho_j\rho_k$, and $\delta_2 = \rho_j$.  
\end{theorem}

\begin{proof}
We define the code $C_{i,(j,k)}$ to be the punctured code from $C(D,B)$ with support $ \{l:Q_l\in \tilde{g}_k^{-1}(\tilde{g}_k(\tilde{g}_j^{-1}(\tilde{g}_j(Q_i))))\}$. This consists of positions corresponding to all points $Q$ in $B$ so that for each $\iota\in\{1,2,\dots, t\}\setminus\{j,k\}$, $y_{\iota}(Q)=y_{\iota}(Q_i)$, and $g(Q)=g(Q_i)$.  We now have that on the set of points corresponding to the positions in the support of $C_{i,(j,k)}$, any function in $V$ restricts to a polynomial in $y_j$ and $y_k$ with bounded degree in each coordinate.  Thus $C_{i,(j,k)}$ is essentially a subcode of a punctured Reed-Muller code on two variables.  We will take the middle code $C_{1_i}=C_{i,(j,k)}$ for each $i\in\{1,2,\dots, n\}$. We observe that 
\[\supp(C_{i,(j,k)})=\bigcup_{l\in \supp(C_{i,j})}\supp(C_{l,k}) = \bigcup_{l\in \supp(C_{i,k})}\supp(C_{l,j}).\] 
Further,  \[\supp(C_{l_1,k})\cap\supp(C_{l_2,k})=\emptyset\] for $l_1\neq l_2$, $l_1,l_2\in \supp(C_{i,j})$.  This is because if $l_1$ and $l_2$ are both in $\supp(C_{i,j})$, then for the corresponding points $Q_{l_1}$ and $Q_{l_2}$ in the evaluation set $B$, we have $y_j(Q_{l_1})\neq y_j(Q_{l_2})$.  If $Q$ is a point with index in the intersection $\supp(C_{l_1,k})\cap\supp(C_{l_2,k})$, then $y_j(Q_{l_1})= y_j(Q)=y_j(Q_{l_2})$, a contradiction. 

Thus every coordinate in $\supp(C_{i,(j,k)})$ has two disjoint recovery sets in $C_{i,(j,k)}$, namely $C_{i,j}$ and $C_{i,k}$, a stronger property than is required for hierarchical recovery.  However, for the purpose of understanding the hierarchical recovery process, we will consider local recovery through the $j$-th recovery set corresponding to $C_{i,j}$ to determine the parameters. This yields $s_2=r_j=\deg(h_j)-\rho_j+1$ and $\delta_2=\rho_j$, where $r_j$ and $\rho_j$ are the local recovery parameters and $d_{h_j}$ is the degree of the covering map $h_j$ from the construction in Theorem \ref{LRCtheorem}.  

The length of $C_{i,(j,k)}$ is $d_{h_j}d_{h_k}$. The dimension of $C_{i,(j,k)}$ is at most $s_1=(d_{h_j}-\rho_j+1)(d_{h_k}-\rho_k+1)$ based on the maximum degrees in $y_j$ and $y_k$ of functions in $V$. The minimum distance of $C_{i,(j,k)}$ can be easily bounded below as follows.  Suppose that there are fewer than $\rho_j\rho_k$ erasures in a codeword of $C_{i,(j,k)}$. Then, there must be some $i^{\prime}$ with fewer than $\rho_k$ erasures in the indices corresponding to $C_{i^{\prime},k}$ or fewer than $\rho_j$ erasures in the indices corresponding to $C_{i^{\prime},j}$. Thus the value of the $i^{\prime}$-th position can be locally recovered by one of $C_{i^{\prime},j}$ or $C_{i^{\prime},k}$. By repeated application, this implies that all erasures can be recovered, so the minimum distance of $C_{i,(j,k)}$ must be at least $\rho_j\rho_k$. 
\end{proof}

\begin{remark}In \cite{ballentine2019codes}, the authors construct codes with hierarchical locality by a natural construction using towers of curves $\X\rightarrow\Y\rightarrow\mathcal{Z}$.  Our hierarchical codes can be viewed in this light, where the tower we construct is $\X\rightarrow \tilde{Y}_j\rightarrow\tilde{Y}_{\{j,k\}}$.  The construction in \cite[Proposition IV.1 ]{ballentine2019codes} requires $\delta_2=2$, while the fiber product construction here allows for a choice of $\delta_2=\rho_j\geq 2$.  Due to the particularly nice arrangement of recovery groups in this construction, the method of recovery using the middle code that we describe here gives a different lower bound $\delta_1$ for the minimum distance of the middle code.
\end{remark}

\subsubsection{Hierarchical recovery of LRC($t$) from a fiber product of Artin-Schreier curves}\label{subsection:ArtinSchreier}
Let $p$ be a prime, $h,t\in \z^+$ with $t\leq h$, and $q=p^{h}$. In \cite{haymaker2018locally, chara2022minimum}, the authors apply the LRC$(t)$ construction to create codes defined over $\F_{q^2}$ on $\mathcal{A}_{q,t}$, a fiber product of $t$ Artin-Schreier curves studied in \cite{van1995construct}. When $t=h$, $\mathcal{A}_{q,t}$ is isomorphic to $\HH_q$.  We let $A=\{a\in\F_{q^2}:a^q+a=0\}$, a $\F_p$-linear space generated over $\F_p$ by $\{a_1, a_2, \dots a_h\}$. 
Let $\Y$ be the projective line in coordinate $y_0$.  For $1\leq i \leq t$, let $\Y_i$ be defined by $y_i-y_i^p = a_iy_0^{q+1}$.   Let $h_i\colon \mathcal{Y}_{i}\rightarrow \mathcal{Y}$ be the map given by projection onto the $y_0$ coordinate.  We may then define $\mathcal{X}=\mathcal{A}_{q,t}$ to be the fiber product of these curves $\mathcal{Y}_{i}$ over $\mathcal{Y}$; i.e.,
\[\mathcal{A}_{q,t}=\mathcal{Y}_{1}\times_{\mathcal{Y}} \mathcal{Y}_{2}\times_{\mathcal{Y}} \dots 
\times_{\mathcal{Y}} \mathcal{Y}_{t}.\] As described in \cite{chara2022minimum}, we may identify $\mathcal{A}_{q,t}$ with its image in $\mathbb{P}^{t+1}$, where the affine points of $\mathcal{A}_{q,t}$ are given by 
      \begin{equation}\label{eq:ASpoints}
        B=\{(y_0,y_1,y_2,\dots, y_t)\in (\mathbb{F}_{q^2})^{t+1}: y_i^p-y_i=a_iy_0^{q+1}\textrm{ for all } 1\leq i\leq t\}.
      \end{equation}
 Let $B$ be the set of $p^tq^2$ affine points in $\X(\mathbb{F}_{q^2})$.
    Let $$P_i=(\alpha,\beta_1,\beta_2,\dots, \beta_t)\in B.$$
    Then $B_{i,j}$, the $j$-th recovery set for the position corresponding to $P_i$, is the set of positions corresponding to the points in $\{(\alpha,y_1,y_2,\dots, y_t)\in B:~ y_k=\beta_k~ \forall ~k\neq j\}.$  We then have $\vert B_{i,j}\vert =p$.
    On points corresponding to the positions in $B_{i,j}$, any function in $V$ varies as a polynomial in $y_j$ of degree at most $(p-2)$ and can therefore be interpolated by knowing its values on any $p-1$ points.

    Given $h, t$ as above, choose $l\leq q^2-\frac{t(p-2)(q+1)p^{t-1}+1}{p^t}$; this will  to ensure  an appropriate value of $d$.  Let $P_{\infty}$ denote the unique point at infinity on $\Y$, and let $D=lP_{\infty}$.
    Then $\mathcal{L}(D)$ is the set of polynomials in $y_0$ of degree at most $l$, a vector space of dimension $\ell(D)=l+1$.

Applying the hierarchical recovery construction in Theorem \ref{H-LRCTheorem} to the codes constructed in Theorem \ref{thm:ASCode}, we immediately obtain an H-LRCs over $\F_{q^2}$.

    \begin{theorem} \label{thm:ASCode}
      Consider $C_{\mathcal{A}_{q,t},l}:=C(V,B)$ where
     $\X=\mathcal{A}_{q,t}$is  the fiber product of the specified Artin--Schreier curves, with $B$ and $l$ as above and $D=l\infty_{\Y}$, and $V$ as in Theorem \ref{LRCtheorem}. Then the $[p^tq^2,(l+1)(p-1)^t,n-lp^t-t(p-2)(q+1)p^{t-1}]$ locally recoverable code $C_{\mathcal{A}_{q,t},l}$ over $\F_{q^2}$ with availability $t$
      and locality $\left(p-1,p-1,\dots,p-1 \right)$ is a $2$-level H-LRC with hierarchical parameters $n_2 = p$, $n_1=p^2$, $s_2=p-1$, $s_1=(p-1)^2$, $\delta_2=2$, $\delta_1=4$.
\end{theorem}

\begin{proof}
The fact that $C_{\mathcal{A}_{q,t},l}$ is an LRC($t$) with the parameter shown follows from \cite{chara2022minimum, haymaker2018locally}. To verify that it is an H-LRC as stated, set $j,k\in\left\{ 1, \dots, t \right\}$, $j\neq k$.  Since the $t$ factor curves are all isomorphic, with degree $p$ maps to the base curve, we have $d_{h_j}=d_{h_k}=p$,  $r_j=r_k=p-1$, and $\rho_j=\rho_k=2$ for any choice of $j,k$. Directly applying Theorem \ref{H-LRCTheorem} gives the required hierarchical parameters.
\end{proof}

\subsubsection{Hierarchical recovery of LRC($2$) from the Hermitian curve as a fiber product}

Codes on the Hermitian curve have been extensively studied, including in the context of locally recoverable codes with availability.  Let $\mathcal{H}_q$ be the Hermitian curve, i.e. the projective curve defined over $\F_{q^2}$ by the affine equation $x^q+x=y^{q+1}$.   The Hermitian curve can be constructed as a fiber product as follows. Let $\Y=\PP^1$, $\Y_1\colon u=y^{q+1}$, and $\Y_2\colon u=x^q+x$. Let $h_j\colon \Y_j\to\PP^1$ be projection onto $u$ for $j=1,2$. 
Note that these maps have coprime degrees. Thus, the corresponding function field extensions are linearly disjoint.
Then the fiber product $\X=\Y_1\times_{\Y}\Y_2$ is isomorphic to the curve $\HH_q\colon x^q+x=y^{q+1}$.
Indeed, the affine points of $\X(\F_{q^2})$ are given by
\[\{((y,u),(x,u)):x,y,u\in \F_{q^2}, y^{q+1}=u=x^q+x\} \subseteq \PP^2\times\PP^2.\]
Hence, this fiber product is isomorphic (by the natural map) to the intersection of the two hypersurfaces in $\PP^3$ with affine equations $u=x^q+x$ and $u=y^{q+1}$ and also to the curve $\HH_q$ defined in $\PP^2$ by affine equation $y^{q+1}=x^q+x$.
Let $C_{\mathcal{H}_q}$ be the LRC($2$) presented in Proposition 5.1 of \cite{barg2015locally}, as well as Theorem 7 of \cite{chara2022minimum}.
    For this code, we take the curve $\HH_q$ with evaluation set $B_{\mathcal{H}_q}=\{P\in\HH_q(\F_{q^2})\colon y(P)\neq 0 \}.$
    We can check that $\vert B_{\mathcal{H}_q}\vert=q^3-q$.  Let $V_{\mathcal{H}_q}$ be the space of functions with basis $\{x^{e_1}y^{e_2}\colon 0\leq e_1 \leq q-2, 0\leq e_2\leq q-1\}$. This choice of functions corresponds to choosing a zero divisor $D=0$ on the base curve $\Y$ in the construction from Theorem \ref{LRCtheorem}, yielding $\ell(D)=1$, and $\rho_1=\rho_2=2$.  Let $C_{\mathcal{H}_q}=C(V_{\mathcal{H}_q}, B_{\mathcal{H}_q})$.
    
    \begin{theorem}[\cite{barg2015locally,chara2022minimum}] \label{thm:hermitian-hierarchy(2)}
            The  $[q^3-q, q^2-q,q^3 -2 q^2 + q +2]$, $((q-2,2),(q-1,2))$ LRC($2$) $C_{\mathcal{H}_q}$ is a $2$-level H-LRC with $n_1=q(q+1)$, $s_1=q(q-1)$, $\delta_1=4$, $n_2=q+1$, $s_2=q-1$, $\delta_2=2$.
    \end{theorem}
    
    \begin{proof}
The supports of the two punctured codes giving recovery sets for the position corresponding to a point $P\in B_{\mathcal{H}_q}$ consist of the positions corresponding to points $Q \in B_{\mathcal{H}_q}$, $Q\neq P$ sharing the same $x$-coordinate as $P$ and those sharing the same $y$-coordinate value as $P$, respectively.  Concretely, if $P_i=(x_0,y_0)$ and $\alpha_i\in\F_q$ such that $x_0^q+x=\alpha_i = y_0^{q+1}$, then $C_{i,1}$ has support $\{j:P_j=(x_0, y)\in B_{\HH_q}, y^{q+1}=\alpha_i\}$ and $C_{i,2}$ has support $\{j:P_j=(x, y_0)\in B_{\HH_q}, \alpha_i=x^q+x\}$.

    Directly applying the construction in the proof of Theorem \ref{H-LRCTheorem} to the fiber product code $C_{\HH_q}$, we obtain an H-LRC with the given hierarchical parameters.  Define $C_{1_i}=C_{i,(2,1)}$ to be the code with support $\{j:P_j=(x, y)\in B_{\HH_q}, y^{q+1}=\alpha_i=x^q+x\}$.  The length of $C_{1_i}$ is $q(q+1)$;  the dimension of $C_{1_i}$ is at most $q(q-1)=s_1$; and the minimum distance of $C_{1_i}$ is at least $4=\delta_1$.  The parameters of the local recovery sets are $s_2=r_2=q-2$ and $\delta_2=\rho_2=2$.  
    \end{proof}
    
     \begin{remark}
   \begin{enumerate} 
   \item  In \cite[Example VII.3]{ballentine2019codes}, the authors create an H-LRC from the Hermitian curve. We obtain a different H-LRC by our fiber product-based construction.
    \item Notice that by simply by considering the other local recovery set, we also have an H-LRC with the less favorable $s_2=r_1=q-2$.  We have chosen the better parameter in the theorem.
    \item In comparing to the standard hierarchy bound of 
    Sasidharan, Agarwal, and Kumar \cite{sasidharan2015codes},  Theorem \ref{thm:hermitian-hierarchy(2)} shows a gap between the bound and the code parameters because the code also has natural availability described in Section \ref{section:hierarchy and availabililty} that is not accounted for in the bound.

     \end{enumerate}
    \end{remark}

We observe that the dimension of the middle code $C_i$ is equal to the dimension of $C_{\HH_q}$. This will occur in the general construction when $t=2$ and $\ell(D)=1$. It is desirable to recover the required erasures by accessing fewer symbols in the middle code than we would have to access in the full code.  Given that the length $n_1$ of $C_{1_i}$ is smaller than the length of $C_{\HH_q}$, giving a higher rate for $C_{1_i}$, one might wonder whether $C_{\HH_q}$ offers any advantages at all.  However, $C_{\HH_q}$ has much larger minimum distance than $C_{1_i}$ promoting greater global error correction or erasure repair. 

This situation motivates discerning additional levels of hierarchy in constructions, the primary topic of the next section.

\section{Codes with $H$-level hierarchy from Reed-Muller and generalized fiber product construction}\label{sec:hlevelH-LRC}

We now present the natural extension of H-LRCs to multiple levels of hierarchy.  This generalization was mentioned in \cite{sasidharan2015codes} and precisely defined in \cite{chen2020cyclic, grezet2021complete}.  Adapted to the  notation above, we have the following definition.

\begin{definition}\label{def:hlevelH-LRC}
    Let $H, n_j, s_j, \delta_j \in\z^+$ for all $j\in [H]$ with $H\geq 2$, $n_{1}>n_{2}>\dots >n_H$, $s_{1}\geq s_2\geq \dots \geq s_H$, and $\delta_{1}>\delta_{2}>\dots >\delta_H$. An $[n,k,d]$ linear code $C$ has $H$-level hierarchical locality with parameters $[(n_1, s_1,\delta_1),(n_2, s_2,\delta_2), \dots, (n_H, s_H,\delta_H)]$ if, for each $i\in\left\{ 1, \dots, n \right\}$, there is a set of $H$ punctured codes $\{(C_j)_i: j\in[H]\}$, each code of length $n_j$, so that $i\in\supp((C_j)_i)$, $I_{j_i}=\supp((C_j)_i)\setminus\{i\}$, so that \begin{itemize}
        \item $(I_H)_i\subseteq (I_{H-1})_i\subseteq\dots\subseteq (I_1)_i$,
        \item $\dim((C_{j})_{i})\leq s_j$ for all $j$,
        \item $(C_j)_i$ has minimum distance at least $\delta_j$ for all $j$,
        \item $(C_j)_i$ has $(H-j)$-level hierarchical locality through the set of codes $\{(C_k)_i: j<k\leq H\}$ for all $k<H-1$.
    \end{itemize}
\end{definition}
The convention is that $n_1<n$, but that is not required. While one could take $C_1=C$, we typically do not do so to highlight the local nature of recovery meaning using a proper subset of coordinates of the code are employed at each level.

We think of the codes $(C_{j})_{1}, \dots, (C_{j})_{n}$ as middle codes for the $j$-th level in the hierarchy for positions $i=1, \dots, n$. Observe that a code with $((n_1, s_1, \delta_1), (n_2,  s_2, \delta_2))$ hierarchical locality
as in Definition \ref{H-LRCdef} has  $H$-level hierarchy with $H=2$. 

As noted in the previous section, some codes with $2$-level hierarchy can be further examined to produced additional tiers for recovery. In this section, we demonstrate that, noting how these recovery sets arise naturally from embedded structures. 

\subsection{$H$-level hierarchy from Reed-Muller codes} Recall that a Reed-Muller code consists of evaluations of points in $\mathbb{F}_q^m$, points which form the finite affine geometry %$EG(m,q)$ 
of dimension $m$. Substructures in this geometry have dimension $m-1, m-2, \ldots, 1$, where the substructures of dimension 1 are lines, those of dimension 2 are planes, and those of dimension $d\geq 3$  are $d$-flats. An $(m-1)$-flat is called a hyperplane. Starting with a codeword position corresponding to a point and choosing a line, then plane, a 3-flat, and so on as the nested recovery sets for that position, the result is an $H$-level H-LRC with $H=m-1$ and hierarchical parameters as described in the next result. 

\begin{theorem}\label{thm:rm-hierarchy}
The Reed-Muller code $\mathcal{RM}_q(v,m)$ is an $(m-1)$-level H-LRC with parameters
\begin{itemize}
    \item $n_j=q^{H-j+1}=q^{m-j}$
    \item $s_j={v+m-j\choose m-j}$
    \item $\delta_j=(q-v)q^{m-j-1}$
\end{itemize} 
for all $j \in \{ 1, \dots, m-1\}$. 
\end{theorem}

\begin{proof} 

    Consider a point $P \in \F_q^m$. To recover an erasure in the coordinate indexed $i$ by $P$, we may consider codes with increasing supports 
    $$
   (I_{m-1})_i \subseteq (I_{m-2})_i \subseteq \cdots \subseteq (I_{2})_i\subseteq (I_1)_i
    $$
    where $I_j \cup \{ P \}$ is a $(m-j)-$dimensional subspace of $\F_q^m$. Then for each $j \in \{1, \dots, m-1\}$, the punctured code $(C_{j})_{i}$
is isometric to the Reed-Muller code $\mathcal{RM}_q(v,m-j)$, which completes the proof.    
\end{proof}

\subsection{$H$-level hierarchy from fiber product codes}
An LRC($t$) arising from a $t$-fold fiber product naturally gives rise to up to $t$-level hierarchy. Figure \ref{Fig:iteratedfiber} is a diagram of the curve covers that give this hierarchical structure.

\begin{figure}
    \centering

\tikzset{every picture/.style={line width=0.75pt}}        

\begin{tikzpicture}[x=0.75pt,y=0.75pt,yscale=-1,xscale=1]
 
\draw    (400,20) -- (450,80) ;
\draw    (400,20) -- (350,80) ;
\draw    (450,105) -- (400,380) ;
\draw    (350,105) -- (400,380) ;

\draw    (350,105) -- (190,185) ;
\draw    (350,105) -- (323,185) ;
\draw    (320,205) -- (400,380) ;
\draw    (190,205) -- (400,380) ;
\draw    (25,305) -- (400,380) ;
\draw    (190,205) -- (120,235) ;
\draw    (95,245) -- (25,275) ;

\draw (380,20) node [anchor= south][inner sep=0.75pt]    {$\X:=\Y_{1} \times _{\Y} \Y_{2} \times _{\Y} \dotsc \times _{\Y} \Y_{t}{}$};
% Text Node
\draw (450,80) node [anchor=north west][inner sep=0.75pt]    {$\Y_{t}$};
% Text Node
\draw (390,80) node [anchor=north east][inner sep=0.75pt]    {$\tilde{\Y}_{t} :=\Y_{1} \times _{\Y} \Y_{2} \times _{\Y} \dotsc \times _{\Y} \Y_{t}{}_{- 1}$};
% Text Node

% Text Node
\draw (400,380) node [anchor=north west][inner sep=0.75pt]    {$\Y$};
% Text Node
\draw (320,180) node [anchor=north][inner sep=0.75pt]    {$\Y_{t-1}$};
% Text Node
\draw (20,180) node [anchor=north west][inner sep=0.75pt]    {$\tilde{\Y}_{\{t,t-1\}} :=\Y_{1} \times _{\Y} \Y_{2} \times _{\Y} \dotsc \times _{\Y} \Y_{t}{}_{- 2}$};
% Text Node
\draw (95,235) node [anchor=north west][inner sep=0.75pt]    {$\dotsc $};
% Text Node
\draw (25,280) node [anchor=north][inner sep=0.75pt]    {$\Y_{1}$};
% Text Node
\draw (280,310) node [anchor=north west][inner sep=0.75pt]    {$\vdots $};

\end{tikzpicture}

    \caption{A diagram of covers in the iterated fiber product construction.  The tower of covers leading to $t$-level hierarchy is the far left sequence of covers.}
    \label{Fig:iteratedfiber}

\end{figure}
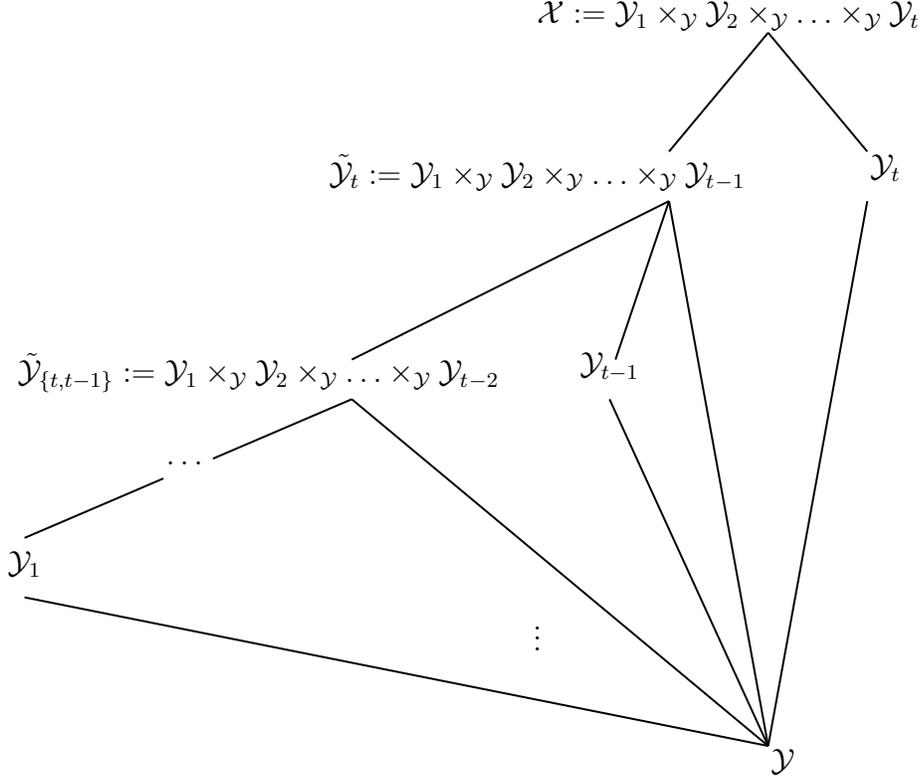

\begin{theorem}\label{H-LRC-H-Level-Theorem}
Let $t\in\z^+$, $t\geq 2$. Let $\X$ be a fiber product satisfying the hypotheses of and with all notation as in Theorem \ref{LRCtheorem}.  Let $C=C_{\X}$ be the $[n,k,d]$, $[(r_j=d_{h_j}-\rho_j+1, \rho_j)]_{j\in\left\{ 1, \dots, t \right\}}$ LRC($t$) constructed in the proof of Theorem \ref{LRCtheorem}.  For $H\in\left\{ 1, \dots, t \right\}$, $2\leq H$, $C$ is an H-LRC with $H$-level hierarchy, and hierarchical parameters:
\begin{itemize}
    \item $n_j=\prod_{k=j}^H d_{h_k}$,
    \item $s_j = \prod_{k=j}^H (d_{h_k}-\rho_k+1)$,
    \item $\delta_j = \prod_{k=j}^H \rho_k$.
\end{itemize}
\end{theorem}

\begin{proof} 
    Recall that $S\subseteq \X(\F_q)$ is the set of evaluation points of $C$.  We define the following two functions to simplify exposition.  Given any set $T\subseteq S$ and $k\in [H]$, let \[F_k(T)=\{\tilde{g}_k^{-1}(\tilde{g}_k(Q)): Q\in T\},\] and \[I(T)= \{i\in\left\{ 1, \dots, n \right\}:Q_i\in T\}.\]

    For any $i\in\left\{ 1, \dots, n \right\}$, let $Q_i\in S$ be the corresponding evaluation point.  We then take \[(S_H)_i=F_H(\{Q_i\}).\]
    This is the set of all evaluation points that share a value of $g$ and all coordinate functions with $Q_i$ except for the $y_{H}$ coordinate.  For $j\in [H-1]$, let \[(S_{H-j})_i=F_{H-j}((S_{H-j+1})_i).\]
    Similarly, the points in set $(S_j)_i$ will share a value of $g$ and all coordinate functions with $Q_i$ except $y_j, y_{j+1}, \dots y_H$.
 
    We then let $(C_j)_i$ be the puncture of the code $C$ to indices in $I((S_j)_i)$. This gives $(I_j)_i=I((S_j)_i)\setminus\{i\}$. The parameters of these codes are now obtainable from simple counting.  The length $n_j$ of the code $(C_j)_i$ will be the product of the degrees of the corresponding maps $h_k$ for $j\leq k \leq H$, because the points of $S$ have been required to fully split in all the corresponding extensions.  Since the evaluation points of the code $(C_j)_i$ share a value of $g$ and all coordinate functions except $y_j, y_{j+1}, \dots y_H$, any function in the original evaluation set $V$ will restrict to a function in the span of the set $\{y_j^{e_j}y_{j+1}^{e_{j+1}}\dots y_H^{e_H}: 0\leq e_k \leq d_{h_k}-\rho_k \textrm{ for all } k\}$.  This yields the given upper bound on the dimension $s_j$ of the code $(C_j)_i$. 
    
    To bound the minimum distance of the code $(C_j)_i$, consider that, by a straight-forward extension of the argument in the proof of Theorem \ref{H-LRCTheorem}, each position $i^{\prime}$ in the support of $(C_j)_i$ has $H-j+1$ disjoint recovery sets in $(C_j)_i$ and the intersection possibilities of any recovery sets for any distinct positions in the support of $(C_j)_i$ is extremely constrained.  Thus we are able to conclude that if there are fewer than $\rho_j\rho_{j+1}\dots\rho_H$ erasures in a codeword of $(C_j)_i$, there must be some position $i^{\prime}$ with fewer than $\rho_k$ erasures in its $k$-th recovery set for some $k\in \{j, j+1,\dots, H\}$.  Therefore position $i^{\prime}$ can be recovered through its $k$-th recovery set; by repeated application any pattern of fewer than $\rho_j\rho_{j+1}\dots\rho_H$ erasures can be recovered.
\end{proof}

\begin{remark}
    In the construction of Theorem \ref{H-LRC-H-Level-Theorem}, we have arbitrarily chosen to use the coordinates $y_1, \dots, y_H$ and fixed an order on these so that the smallest code has evaluation points varying in $y_H$, intermediate middle codes each incorporate variability in previous coordinates, and the largest middle code has evaluation points varying in $y_1$ through $y_H$.  Given any choice of $H$ values in $\left\{ 1, \dots, t \right\}$ and any order of these values, we could use the fact that the fiber product construction is (up to isomorphism) commutative and associative to reorder the factors of the fiber product so that our chosen factors are in the positions $1$ through $H$ in the fiber product. Thus, this arbitrary choice of coordinates does not result in any loss of generality.
\end{remark}

\begin{remark}In \cite[Section VIII Part D]{ballentine2019codes}, the authors create codes with hierarchical locality from coverings of algebraic curves, including a construction involving fiber products which is substantially different than that described here. The fiber product construction of Ballentine, Barg, and Vladut uses curves $\mathcal X, \mathcal Y, \mathcal Z, \mathcal C$, where there exist covering maps $\phi_2:\mathcal X\rightarrow \mathcal Y$ and $\phi_1: \mathcal Y\rightarrow \mathcal Z$ with $\deg(\phi_2)=ab$ and $\deg(\phi_1)=b$, and both $\mathcal X\times_{\mathcal Z} \mathcal C$ and $\mathcal Y\times_{\mathcal Z} \mathcal C$ are  smooth, absolutely irreducible curves.  Thus the construction starts with a tower of curves $\mathcal X\rightarrow \mathcal Y \rightarrow \mathcal Z$ and then takes the fiber products (over $\mathcal Z$) of each of these curves with another curve $\mathcal C$ that is also a cover of $\mathcal Z$.  The code is then defined over $\mathcal X\times_{\mathcal Z} \mathcal C$, with hierarchy resulting from covers in the tower $\mathcal X \times_{\mathcal Z} \mathcal C\rightarrow \mathcal Y\times_{\mathcal Z} \mathcal C \rightarrow \mathcal Z\times_{\mathcal Z}  \mathcal C\cong \mathcal C$.  In this work, we pursue a more general perspective on fiber product codes based on the construction in \cite{haymaker2018locally}. One major difference is that the iterated fiber product in \cite{haymaker2018locally} starts with some collection of curves $\mathcal X_1, \mathcal X_2, \dots \mathcal X_t$, with no covering relationships among them but with each having a map to a shared curve $\mathcal Y$. We consider a code defined on
\[\mathcal X=\mathcal X_1\times_{\mathcal Y} \mathcal X_2 \times_{\mathcal Y} \dots \times_{\mathcal Y} \mathcal X_t.\] The underlying fiber product of Ballentine et al. in  \cite{ballentine2019codes} may coincide with this construction in some very special cases, but most situations of this paper and \cite{haymaker2018locally} are not covered by \cite{ballentine2019codes}.  First, our construction does not require the strong restriction on the map degrees in \cite{ballentine2019codes}.  Further, the highest level of hierarchy in our construction arises from the pullback of the map $\mathcal X\rightarrow \mathcal Y$, where $\mathcal Y$ is the base curve of every fiber product.  The base curve of the fiber product is not represented in the final tower of \cite{ballentine2019codes}.  Also, the construction in \cite{ballentine2019codes} begins with a two-covering tower of curves and uses fiber products to obtain a related two-covering tower of curves.  The construction in \cite{haymaker2018locally} that is continued here begins with many single-level covers of curves and uses each of these to create an additional level of hierarchy, with $H$ levels of hierarchy arising from $H$ single-level covers.  Though not explicitly stated, it seems that the fiber product construction in \cite{ballentine2019codes} would need to start with an $H$-level tower of curves to generalize to $H$-level hierarchy after the fiber product is applied.
\end{remark}

We return to the curve $\mathcal{A}_{q,t}$ described in Subsection \ref{subsection:ArtinSchreier}. Let $p$ be an odd prime, $t,h\in\mathbb{Z}^{+}$, $q=p^h$, $t\leq h$. Choose $l\in\mathbb{Z}$ as before to define an LRC($t$) $C_{\mathcal{A}_{q,t},l}$ over $\mathbb{F}_{q^2}$ with parameters given in Theorem \ref{thm:ASCode}.  This code has $H$-level hierarchy for any $H\leq t$. To display the full range of hierarchy, we choose $H=t$.
\begin{theorem}
The code $C_{\mathcal{A}_{q,t},l}$ constructed in Theorem \ref{thm:ASCode} is a $t$-level H-LRC with hierarchical parameters 
\begin{itemize}
    \item $n_j=p^{t+1-j}$,
    \item $s_j=(p-1)^{t+1-j}$, and
    \item $\delta_j=2^{t+1-j}$,
\end{itemize}
for $1\leq j\leq t$.
\end{theorem}
\begin{proof}
The parameters are a direct application of Theorem \ref{H-LRC-H-Level-Theorem}, where our construction gives $\rho_k=2$ and $d_{h_k}=p$ for all $k$.
\end{proof}

\section{Combined hierarchy and availability }
\label{section:hierarchy and availabililty}

Finally, we incorporate the notion of availability into our study of hierarchical locality.  The following definition is a generalization of those given by Freij-Hollanti, Westerback, and Hollanti in \cite{freij2016locally} and Ballentine, Barg, and Vladut in \cite{ballentine2019codes}.  

\begin{definition} \label{def:hlevelH-LRCt}
    Let $H, n_j, t_j, s_j, \delta_j, r_{j,k}, \rho_{j,k}, \in\z^+$ for all $j\in [H]$ and $k\in [t_j]$ with $H\geq 2$, $s_{1}\geq s_{2}\geq \dots \geq s_{H}$ and $\delta_{1}>\delta_{2}>\dots >\delta_{H}$. 
Let $C$ be an $[n,k,d]$ linear code with $H$-level hierarchical locality with parameters  $(n_j, s_{j},\delta_{j})_{j\in [H]}$. For every $i\in \left\{ 1, \dots, n \right\}$, $j\in [H]$, let $(C_{j})_i$ be the $j$-th level middle code for position $i$. If, for all $i$, $(C_{j})_i$ is a locally recoverable code with  availability $t_j$ and local parameters $((r_{j,k},\rho_{j,k})_{k\in [t_j]})$, we say that $C$ has $H$-level hierarchy with availability $(t_j)_{j\in[H]}$.
\end{definition}

The notion of H-LRCs with availability is useful  in capturing the additional flexibility of the geometric constructions we present here.  Intuitively, if a middle code has availability, it potentially offers the ability to recover many erasure patterns in the middle code using local recovery.  There may be another erasure recovery algorithm for the middle code, but having availability means we have additional local strategies at our disposal.

We now discuss the parameters of our example families when viewed in the more structured framework as codes with both $H$-level hierarchy and availability.  Definition \ref{def:hlevelH-LRCt} describes codes with an $H$-level hierarchical recovery structure and $t_j$ linearly independent recovery sets at level $j$, $j\in [H]$. 

\subsection{Reed-Muller codes with hierarchy and availability} At each hierarchy level of the Reed-Muller code, there is natural availability using the underlying geometry of the affine spaces of $\mathbb{F}_q^m$. 

\begin{corollary}\label{cor:RM-LRCt} 
    Let $C$ be the hierarchical locally recoverable code in Theorem~\ref{thm:rm-hierarchy}. Then $C$ is also an $H=(m-1)$-LRC with availability parameters $t_j=\frac{q^{m+1-j}-1}{q-1}$, for $j=1, \ldots, m-1$, and $(r_{j,k}, \rho_{j,k})=(r_k, \rho_k)$ for $k\in [t_j]$ and $j\in [H]$. 
\end{corollary}

\begin{remark}
    For a fixed point $P\in \mathbb{F}_q^m$ that corresponds to the codeword position of $C$ that we seek to recover, any pair of  $(m-j)$-flats that contain $P$  must overlap in an affine subspace of dimension $(m-j-1)$. Suppose $A_1, A_2$ are two affine subspaces of dimension $(m-j)$ with $P\in A_1$, $P\in A_2$;  then $B=A_1\cap A_2$ is an affine subspace of dimension $(m-j-1)$.  We note that the conditions of  Definition~\ref{def:hlevelH-LRCt} are still satisfied by the code $C$ since there is a recovery method using only the points in $A_1\setminus B$ (resp., $A_2\setminus B$), as follows. To recover the codeword value at the point $P$, consider a line $L$ through $P$ such that $L\setminus P\subseteq A_1\setminus B$. Notice that $|(L\setminus P)\cap (A_1\setminus B)|=q-1$, since otherwise the entire line would be contained in the affine subspace $B$. Thus, since $v<q$ the value of a degree $v$ polynomial can be recovered from the evaluation of the polynomial on the points on $L\setminus P$. 
\end{remark}

\subsection{Fiber product codes with hierarchy and availability}
For codes from fiber products, we get natural availability at each level. 

\begin{corollary}\label{thm:FiberProductH-LRCt}
Let $t\in\z^+$, $t\geq 2$. Let $\X$ be a fiber product satisfying the hypotheses of and with all notation as in Theorem \ref{LRCtheorem}.  Let $C=C_{\X}$ be the $[n,k,d]$, $[(r_j=d_{h_j}-\rho_j+1, \rho_j)]_{j\in\left\{ 1, \dots, t \right\}}$ LRC($t$) constructed in the proof of Theorem \ref{LRCtheorem}. Let $H\in\left\{ 1, \dots, t \right\}$, $2\leq H$, and consider the code $C$ described as an H-LRC in Theorem \ref{H-LRC-H-Level-Theorem}.  Then $C$ is an H-LRC with $H$-level hierarchy and hierarchical parameters as given in the theorem, with availability parameters $t_j=H+1-j$ for $j\in[H]$ and $(r_{j,k},\rho_{j,k})=(r_k,\rho_k)$ for $k\in[t_j]$ and $j\in[H]$.

\end{corollary}

\begin{remark}
    As in the case of simple hierarchy, the authors of \cite{ballentine2019codes} devise H-LRCs with availability from a certain $2$-level fiber product (see Figure 2 in the cited paper).  Their construction is a specialization of their main construction of hierarchy from a two-cover tower to a situation where each cover can be decomposed such that the corresponding covering curve is expressed as a fiber product over the covered curve. Our construction is not captured by \cite{ballentine2019codes}.  Our construction uses fiber products differently. One very concrete difference is that our construction does not ever involve a fiber product using another fiber product as a base curve.  Our hierarchy comes from the fibers of maps in a tower of fiber products with an increasing number of factors, with all fiber products using $\Y$ as a base curve, as in Figure \ref{Fig:iteratedfiber}.
\end{remark}

\section{Conclusion} \label{conclusion}

In this paper, we harnessed the underlying structure of Reed-Muller and fiber product codes to provide hierarchical local recovery of erasures. H-LRCs have the advantage that they make use of smaller recovery sets for larger numbers of erasures than LRCs without tiered recovery sets of varying sizes. The constructions considered here are especially useful as their properties rely on the underlying geometry and immediately satisfy the disjoint or linearly independent repair group property. 
It remains to consider hierarchical recovery when disjoint or linearly independent repair groups are not required and how that may yield additional recovery sets at each level. 

\backmatter 
\bmhead{Acknowledgments}
The National Science Foundation partially supported the second author (DMS-2137661) and the third author (DMS-2201075). The third author is also partially supported by the Commonwealth Cyber Initiative.

\bibliography{HLRC-biblio}

 \end{document}